\documentclass[journal, twoside]{IEEEtran}
\usepackage{hyperref,graphicx}
\usepackage{amsmath}
\usepackage{amssymb}
\usepackage{amsthm}
\usepackage{bm} 
\usepackage{dirtytalk} 
\usepackage{booktabs} 
\usepackage{multirow}
\usepackage{booktabs}
\usepackage{cleveref}
\usepackage{cite}
\usepackage{xcolor}
\usepackage{lipsum}
\newcommand{\bbR}{\mathbb{R}}
\newcommand{\bbC}{\mathbb{C}}
\newcommand{\bbH}{\mathbb{H}}
\newcommand{\bbCi}{\mathbb{C}_{\bm i}}
\newcommand{\bbCj}{\mathbb{C}_{\bmj}}

\newcommand{\domega}{\mathrm{d}\nu}
\newcommand{\dnu}{\mathrm{d}\nu}
\newcommand{\bmj}{\bm j}
\newcommand{\bmi}{\bm i}
\newcommand{\bmk}{\bm k}
\newcommand{\bmmu}{\bm \mu}

\renewcommand{\epsilon}{\varepsilon}
\renewcommand{\phi}{\varphi}

\newcommand{\defeq}{:=}

\newcommand{\ie}{\emph{i.e. }}
\newcommand{\Expe}[1]{\mathbf{E}\left\lbrace #1\right\rbrace}
\newcommand{\dX}{\mathrm{d}X}
\newcommand{\dU}{\mathrm{d}U}
\newcommand{\dV}{\mathrm{d}V}
\newcommand{\involj}[1]{#1^{\star\bmj}}
\newcommand{\involmu}[1]{#1^{\star\bmmu}}
\newcommand{\involi}[1]{#1^{\star\bmi}}
\newcommand{\involk}[1]{#1^{\star\bmk}}
\newcommand{\polarmodj}[1]{\left\vert #1 \right\vert_{\bmj}^2}
\newcommand{\Span}[1]{\mathrm{span}\left\lbrace #1\right\rbrace}

\newcommand{\intnu}{\int_{-1/2}^{+1/2}}

\newtheorem{theorem}{Theorem}
\newtheorem{proposition}{Proposition}

\begin{document}

\title{Spectral analysis of stationary\\ random bivariate signals}

\author{Julien Flamant, Nicolas Le Bihan and Pierre Chainais}

\author{Julien~Flamant,~\IEEEmembership{Student Member,~IEEE,}
        Nicolas~Le Bihan,
        and~Pierre~Chainais,~\IEEEmembership{Senior~Member,~IEEE}
\thanks{J. Flamant and P. Chainais are with Univ. Lille, CNRS, Centrale Lille, UMR 9189 - CRIStAL - Centre de Recherche en Informatique Signal et Automatique de Lille, 59000 Lille, France. N. Le Bihan is with CNRS/GIPSA-Lab, 11 Rue des math\'{e}matiques, Domaine Universitaire, BP 46, 38402 Saint Martin d'H\`{e}res cedex, France. Part of this work has been funded by the CNRS, GDR ISIS, within the SUNSTAR interdisciplinary research program.}}%

\markboth{}%
{Flamant, Le Bihan and Chainais: Spectral analysis of stationary random bivariate signals}

\maketitle

\begin{abstract}
	A novel approach towards the spectral analysis of stationary random bivariate signals is proposed. Using the Quaternion Fourier Transform, we introduce a quaternion-valued spectral representation of random bivariate signals seen as complex-valued sequences. This makes possible the definition of a scalar quaternion-valued spectral density for bivariate signals. This spectral density can be meaningfully interpreted in terms of frequency-dependent polarization attributes. A natural decomposition of any random bivariate signal in terms of unpolarized and polarized components is introduced. Nonparametric spectral density estimation is investigated, and we introduce the polarization periodogram of a random bivariate signal. Numerical experiments support our theoretical analysis, illustrating the relevance of the approach on synthetic data. 
\end{abstract}

\section{Introduction}
	\IEEEPARstart{R}{andom} bivariate signals are 2D vector timeseries. They appear in a large variety of applications, ranging  from oceanography \cite{gonella1972rotary,mooers1973technique}, to optics \cite{erkmen2006optical}, radar \cite{Ahrabian2013}, geophysics \cite{samson1983pure} or EEG analysis \cite{sakkalis2011review} to name but a few. A bivariate signal is usually decomposed in two orthogonal components $u[t]$ and $v[t]$. Thus a bivariate signal $x[t]$ can be either represented as the vector signal $x[t] = (u[t],\: v[t])^T \in \bbR^2$ or the complex-valued signal $x[t] = u[t] + \bmi v[t]$. 
	
	The statistical analysis of signals with vector-valued samples can be carried out using standard multivariate time series analysis techniques (see \emph{e.g} \cite[chap. 9]{priestley1981spectral} or \cite[chap. 11]{brockwell2013time}), bivariate signals are no exception. However, in the signal processing community, bivariate signals have often been described using complex-valued models \cite{Picinbono1997b,POA96I,POA96II,schreier2010statistical,Adali2011}. To account for the full second-order statistical characterization of the complex signal $x[t] = u[t] + \bmi v[t]$, a usual approach is to define two quantities: the usual autocovariance function and the complementary-covariance function (the relation function in \cite{Picinbono1997b}). This leads to the definition of the augmented vector $(x[t], \overline{x[t]})^T \in \bbC^2$ from the signal and its conjugate, and to the related augmented covariance and spectral density matrices \cite[chap. 8]{schreier2010statistical}.

	The {\em rotary spectrum} analysis \cite{gonella1972rotary,mooers1973technique} is a well-known technique rooted in oceanographic studies. It is based on the decomposition of a complex-valued signal into clockwise and counterclockwise rotating components. This seminal approach has stimulated many theoretical developments \cite{Schreier2008,walden2013rotary,Walden2007, Rubin-Delanchy2008,Chandna2011,Chandna2013}.
	The rotary components method is also related to polarization analysis \cite{Schreier2008, schreier2010statistical}. The focus on polarization or rotary components usually depends on the field of application: rotary components are more common in oceanography, while optics and radar scientists usually deal with polarization \cite{brosseau1998fundamentals,Brosseau1995statistics,Barakat1987,Barakat1985,Giuli1986}.


	Existing approaches all rely on the same spectral representation of stationary bivariate processes, which is based on the standard Fourier Transform (FT). For complex signals, it means that negative frequencies must be taken into account, as they provide information about the process. Therefore one needs to consider spectral matrices rather than a scalar spectral density. As a consequence, meaningful physical parameters are not directly ``readable'' in the state-of-the-art formulations. 
	

	We propose a new approach to analyze the spectral content of stationary random bivariate signals. It is based on recent results from \cite{1609.02463,flamant2017polarization} and extended to the case of stationary bivariate signals seen as complex-valued signals. This paper provides a well-suited framework for the analysis of stationary bivariate signals which naturally describes the spectral content and the ``geometric'' (also sometimes called polarization) content of bivariate signals. Thanks to the definition of the dedicated Quaternion Fourier Transform (QFT), it is possible to describe the spectral content of such signals in terms of polarized and unpolarized parts, which both encode meaningful information about the signal.


	This paper structure is as follows. Section \ref{section:quaternionSec2} reviews the necessary material regarding quaternions and the QFT. Section \ref{sec:spectralRep} gives the central results of this paper: we introduce the scalar quaternion-valued spectral density of a bivariate signal and its subsequent properties. Results are compared with state-of-the art approaches. In particular the differences between second-order circularity, also called properness, and polarization are stressed. Simple explicit examples are finally presented. Section \ref{sec:spectral_density_estimation} deals with nonparametric spectral density estimation, and introduces the polarization periodogram. Our theoretical analysis is supported by numerical experiments in Section \ref{sec:examples}. Section \ref{sec:conclusion} gathers concluding remarks.


\section{Quaternion Fourier Transform}\label{section:quaternionSec2}

	\subsection{Quaternion algebra}\label{section:quaternionAlgebra}

	We review the basic material regarding quaternions and refer to more detailed textbooks (\emph{e.g.} \cite{conway2003quaternions}) for a complete overview. Quaternions form a four dimensional noncommutative algebra. Any quaternion $q\in \bbH$ can be written in its Cartesian form as
	\begin{equation}\label{eq:Cartesianformquaternion}
		q = a + b\bmi + c\bmj + d\bmk,
	\end{equation}
	where $a, b, c, d \in \bbR$ and $\bmi, \bmj, \bmk$ are roots of $-1$ satisfying
	\begin{equation}
		\bmi^2 = \bmj^2 = \bmk^2 = \bmi\bmj\bmk = -1.
	\end{equation}
	The canonical elements $\bmi, \bmj, \bmk$, together with the identity of $\bbH$ form the quaternion canonical basis given by $\left\lbrace 1, \bmi, \bmj, \bmk\right\rbrace$. We will use the notation $\mathcal{S}(q) = a \in \bbR$ to define the \emph{scalar part} of the quaternion $q$, and $\mathcal{V}(q) = q - \mathcal{S}(q) \in \Span{\bmi, \bmj, \bmk}$ to denote its \emph{vector part}. We can define the real and imaginary parts of a quaternion $q$ as $\mathfrak{R}(q) = a, \: \mathfrak{I}_{\bmi}(q) = b, \:  \mathfrak{I}_{\bmj}(q) = c, \: \mathfrak{I}_{\bmk}(q) = d.$
	A quaternion is called \emph{pure} if its real (or scalar) part is equal to zero, that is $a = 0$, \emph{e.g.} $\bmi, \bmj, \bmk$ are pure quaternions. The quaternion conjugate of $q$ is $\overline{q} = \mathcal{S}(q) - \mathcal{V}(q)$. The modulus of a quaternion $q \in \bbH$ is defined by $\vert q \vert^2 = q\overline{q} = \overline{q}q = a^2+b^2+c^2+d^2$. The inverse of a non-zero quaternion is defined by $q^{-1} = \overline{q}/\vert q\vert^2$. Importantly quaternion multiplication is noncommutative, that is in general for $p, q \in \bbH$, one has $pq \neq qp$. 
	Involutions with respect to $\bmi, \bmj, \bmk$ are defined as $\overline{q}^{\bmi} = - \bmi q \bmi, \: \overline{q}^{\bmj} = - \bmj q \bmj, \: \overline{q}^{\bmk} = - \bmk q \bmk$.
	The combination of conjugation and involution with respect to an arbitrary pure quaternion $\bmmu$ is denoted by $\involmu{q} \defeq \overline{\left(\overline{q}\right)}^{\bmmu} = \overline{\left(\overline{q}^{\bmmu}\right)}$ and for instance $\involj{( a + b\bmi + c\bmj + d\bmk)} = a + b\bmi - c\bmj + d\bmk$. For later use, we also introduce the notation $\polarmodj{q} \defeq q\involj{q}$.

	Quaternions encompass complex numbers. One can construct \emph{complex subfields} of $\bbH$, \emph{e.g…} $\bbCj = \Span{1, \bmj}$ or $\bbCi = \Span{1, \bmi}$ which are isomorphic to $\bbC$. Any quaternion can be seen as a pair of complex numbers: let us mention the symplectic decomposition $q = q_1 + \bmi q_2, \: q_1, q_2 \in \bbCj$, where the quaternion $q$ is splitted into two $\bbCj$-valued complex numbers. This form is particularly suited for computations performed later on with the quaternion Fourier transform. 

	Polar forms of quaternions exist. For an arbitrary pure unit quaternion $\bmmu$ and $\theta \in \bbR$, we have $\exp(\bmmu\theta) = \cos\theta + \bmmu\sin\theta$. It generalizes the notion of complex exponentials and the following polar form was proposed in \cite{bulow2001hypercomplex}:
	\begin{equation}\label{eq:polarForm}
		q = \vert q \vert\exp[\bmi \theta]\exp[-\bmk\chi]\exp[\bmj\phi],
	\end{equation}
	with $(\theta, \chi, \phi) \in [-\pi/2, \pi/2]\times [-\pi/4, \pi/4]\times [-\pi, \pi]$. This form is particularly useful for quaternion embedding of complex signals, see \cite{1609.02463}, and in the spectral description of stationary signals in Section \ref{sub:examplesSpecRep}.
	
	\subsection{Quaternion Fourier Transform}
	We review here briefly the Quaternion Fourier Transform (QFT) introduced in \cite{le2014instantaneous} and studied in detail recently in \cite{1609.02463}. We refer the reader to these articles for proofs and a detailed presentation.

	Here we consider only discrete-time (DT) signals: $t$ is a time index such that $x(t\Delta) = x[t]$, where $\Delta$ is the sampling step. We assume $\Delta = 1$ in the rest of this paper. 

	We define the Discrete Time Quaternion Fourier Transform (hereafter denoted QFT) of \emph{axis} $\bmj$ of a signal $x: \bbR \rightarrow \bbH$ by
	\begin{equation}\label{eq:definitionQFT}
		X(\nu) \defeq \sum_{t=-\infty}^{+\infty} x[t]\exp(-\bmj2\pi\nu t).
	\end{equation}
	The inverse QFT is given by
	\begin{equation}\label{eq:inverseQFT}
		x[t] = \intnu X(\nu)\exp(\bmj2\pi\nu t)\mathrm{d}\nu.
	\end{equation}
	The above relations are directly obtained by discretizing the continuous-time QFT presented in \cite{1609.02463,le2014instantaneous} using similar arguments as with the usual FT.

	Write $x[t] = u[t] + \bmi v[t]$, $u[t], v[t] \in \bbCj$, its QFT reads
	\begin{equation}
		X(\nu) = U(\nu) + \bmi V(\nu), \quad U(\nu), V(\nu) \in \bbCj,
	\end{equation}
	where $U(\nu), V(\nu)$ are the standard FTs of $u, v$: the QFT is performing two standard FT. This may explain why this QFT shares most properties of the classical FT, see \cite{1609.02463}. 

	Another benefit of the QFT is that it provides a well-suited framework for bivariate signals. A bivariate signal can be written as a $\bbCi$-valued signal $x[t] = u[t] + \bmi v[t]$, with $u, v$ real signals. The QFT of such a signal exhibits an $\bmi$-hermitian symmetry \cite{le2014instantaneous}:
	\begin{equation}\label{eq:ihermitian}
		X(-\nu) = \overline{X(\nu)}^{\bmi}.
	\end{equation}
	Eq. (\ref{eq:ihermitian}) shows that, when using the QFT with $x[t] \in \bbCi$, negative frequencies carry no additional information to positive frequencies about the signal. 


	The $\bmi$-hermitian symmetry (\ref{eq:ihermitian}) permits the construction of the \emph{quaternion embedding of a complex signal}, by canceling out negative frequencies of the spectrum. The quaternion embedding of a complex signal is a direct bivariate counterpart of the usual analytic signal and permits to identify both instantaneous phase and polarization (\ie geometric) properties of a complex signal.
\section{Spectral representation of bivariate stationary random processes}\label{sec:spectralRep}


	Any bivariate discrete-time random process $x[t]$ can be decomposed as $x[t] = u[t] + \bmi v[t]$, where $u[t], v[t]$ are real-valued discrete-time random processes. The process $x[t]$ is said to be \emph{second-order stationary} if $u[t]$ and $v[t]$ are \emph{jointly second-order stationary}, that is \cite[p. 655]{priestley1981spectral}:
	\begin{align}
		\Expe{x[t]} &= \Expe{u[t]} + \bmi \Expe{v[t]}= m \in \bbCi,\\
		R_{uu}[t, \tau] &= \Expe{u[t]u[t-\tau]} = R_{uu}[\tau], \: R_{uu}[0]< \infty,\\
		R_{vv}[t, \tau] &= \Expe{v[t]v[t-\tau]} = R_{vv}[\tau], \: R_{vv}[0]< \infty,\\
		R_{uv}[t, \tau] &= \Expe{u[t]v[t-\tau]} = R_{uv}[\tau].
	\end{align}
	Here $\Expe{\cdot}$ denotes the mathematical expectation and $R_{uu}, R_{vv}$ and $R_{uv}$ denote usual autocovariance sequences (ACVS) and crosscovariance sequences (CCVS) between real-valued sequences. Second-order stationarity ensures that first and second-order moments are finite and do not depend on $t$.

	Let us make some further assumptions. First all processes are assumed to be zero-mean, \ie $m = 0$. Also, we suppose that any of the ACVS or CCVS are absolutely summable, so that their usual Fourier transform exist. Second-order stationarity is simply referred to as \emph{stationarity} in the sequel.

	\subsection{Main result}

	Using the QFT, we derive a spectral representation theorem for any bivariate stationary process $x[t] = u[t] + \bmi v[t]$. The existence of the spectral increments $\mathrm{d}X(\nu)$ follows from the existence of the (usual) spectral increments of $u[t]$ and $v[t]$, see for instance \cite[p. 36]{Yaglom1962}, \cite[p. 246]{priestley1981spectral} or \cite[p. 344]{Blanc-Lapierre1953}.
Namely, we can write $x[t]$ as the (quaternion) Fourier-Stieltjes integral
	\begin{equation}
		x[t] = \intnu \mathrm{d}X(\nu)\exp(\bmj 2\pi\nu t),
	\end{equation}
 	where the spectral increments $\mathrm{d}X(\nu) = X(\nu + \dnu) - X(\nu)$ are quaternion-valued and $X(\nu)$ is an independent additive random measure. 
	The existence roots in two main properties. The QFT restricted to $\bbCj$-valued processes is isomorphic to the standard FT so that the spectral increments $\mathrm{d}U(\nu), \mathrm{d}V(\nu)$ of $u[t], v[t]$ are $\bbCj$-valued. Moreover the QFT is left-quaternion-linear, that is $\forall \lambda \in \bbH$, the QFT of $\lambda x[t]$ is $\lambda X(\nu)$ so that
 	\begin{equation}
 		\mathrm{d}X(\nu) = \mathrm{d}U(\nu) + \bmi\mathrm{d}V(\nu).
 	\end{equation}
 	

	\begin{theorem}[Spectral representation of bivariate stationary random processes]\label{thm:spectralRep}
		Let $x[t] = u[t] + \bmi v[t]$ be a bivariate stationary process. Suppose that $u[t]$ and $v[t]$ are both harmonizable. Then there exists a quaternion-valued orthogonal process $X(\nu)$ such that
		\begin{equation}\label{eq:spectralRepCT}
			x[t] = \intnu \mathrm{d}X(\nu)\exp(\bmj 2\pi\nu t),
		\end{equation}
		the integral being defined in the mean-square sense. The process $X(\nu)$ has the following properties:
		\begin{enumerate}
			\item $\forall \nu,\: \Expe{\dX(\nu)} = 0$, 

			\item $\forall \nu, \: \Expe{\vert\dX(\nu)\vert^2} + \Expe{\polarmodj{\dX(\nu)}} \bmj = \Gamma_{xx}(\nu)\dnu$, where $\Gamma_{xx}(\nu)$ is the spectral density of $x$,


			\item For any $\nu \neq \nu'$, we have
			$$\Expe{\dX(\nu)\overline{\dX(\nu')}} = \Expe{\dX(\nu)\involj{\dX(\nu')}} = 0,$$
			which shows that the spectral increments $dX(\nu)$ are two times orthogonal.
		\end{enumerate}
	\end{theorem}

	Appendix \ref{app:spectralRepProof} proves this theorem that follows the derivation given by Priestley \cite{priestley1981spectral}, adapted to the QFT setting. Properties 1), 2) and 3) essentially come from the self and joint properties of the spectral increments $\mathrm{d}U(\nu)$ and $\mathrm{d}V(\nu)$. Property 2) introduces the quaternion-valued spectral density of $x$. This is in fact a \emph{power spectral density} since (see Appendix \ref{app:spectralRepProof}) one has
	\begin{equation}\label{eq:powerSpectralDensity}
		\intnu \Gamma_{xx}(\nu)\dnu = \Expe{\vert x[t]\vert^2}+ \Expe{\polarmodj{x[t]}}\bmj \in \bbH
	\end{equation}
	where the right-hand side contains all the power information of the process $x[t]$. 

	Note that four (real) power-related quantities are necessary to describe the second-order properties of a bivariate signal, see Section \ref{ssec:spectralDensityStokes}. This argues in favor of the definition of relevant scalar quaternion-valued quantities such as $\Gamma_{xx}(\nu)$.

	Note also that in our definition $\Gamma_{xx}(\nu)$ is a generalized function: we can use Dirac distributions. Thus $\Gamma_{xx}(\nu)$ does not need to be continuous everywhere.

	Since $x[t] = u[t] + \bmi v[t]\in \bbCi$, the spectral increments additionally satisfy the same $\bmi$-hermitian symmetry as the QFT of deterministic $\bbCi$-valued signals, \ie
	\begin{equation}\label{eq:ihermitianSpectralIncrements}
		\dX(-\nu) = \overline{\dX(\nu)}^{\bmi}.
	\end{equation}
	As a result, the spectral density $\Gamma_{xx}(\nu)$ has symmetry
	\begin{equation}\label{eq:isymmetryGammaxx}
		\Gamma_{xx}(-\nu) = \involi{\Gamma_{xx}(\nu)}.
	\end{equation}
	This result shows again that the study of $\bbCi$-valued (bivariate) signals can be performed using only positive frequencies of its quaternion-valued spectral representation. At each (positive) frequency, a quaternion-valued quantity summarizes both power and polarization properties of $x$. This will be detailed in Sections \ref{ssec:spectralDensityStokes} and \ref{sub:Poincaresphere}.



	The spectral increments $\dX(\nu)$ are quaternion-valued random variables (RV). 
	It is usual to describe the full second-order statistical structure of a quaternion RV $q$ by the four covariances $\Expe{q\involmu{q}}$, $\bmmu = \bmi, \bmj, \bmk$.
	These covariances often obey some symmetries characterized by the notion of {\em properness}. Properness levels of quaternion RV have been investigated by several authors \cite{vakhania1999random,amblard2004properness,via2010properness} and reviewed recently in \cite{lebihanproperness15}.  
	The spectral increments of a $\bbCi$-valued process $x[t]$ satisfy the symmetry (\ref{eq:ihermitianSpectralIncrements}) and thus property 3) of theorem \ref{thm:spectralRep} with $\nu' = -\nu$ yields
	\begin{equation}\label{eq:propernessSpectralIncrements}
		\Expe{\dX(\nu)\involi{\dX(\nu)}}= \Expe{\dX(\nu)\involk{\dX(\nu)}} = 0.
	\end{equation} 
	Eq. (\ref{eq:propernessSpectralIncrements}) shows that the spectral increments $\dX(\nu)$ are $(1, \bmj)$-proper in the classification of \cite{lebihanproperness15}, also denoted as $\bbCj$-properness in \cite{amblard2004properness}.
	


	\subsection{Covariances, Wiener-Khintchine theorem}
	
	Thanks to theorem~\ref{thm:spectralRep}, we are able to describe the spectral content of random bivariate signals by their spectral density. Usually, this spectral density is introduced by using Wiener-Khintchine theorem once the autocorrelation of the process has been defined. This is not the case here due to the non-commutativity of $\bbH$: the notion of (auto-)covariance must be carefully defined if one wants to recover a Wiener-Khintchine theorem for quaternion valued processes. To this aim, a natural approach can be to define the autocovariance of $x$ by inverse QFT of $\Gamma_{xx}(\nu)$:
	\begin{equation}
		\label{gamma_xx_WK}
		\gamma_{xx}[\tau] = \int_{-1/2}^{+1/2} \Gamma_{xx}(\nu)\exp(\bmj2\pi\nu \tau)\mathrm{d}\nu,
	\end{equation}
 	so that $\gamma_{xx}[\tau]$ is explicitly given by (see Appendix \ref{app:spectralRepProof})
 	\begin{equation}\label{eq:gamma_xxAutocorrDef}
 	\begin{split}
 		\gamma_{xx}[\tau] &= R_{uu}[\tau] + R_{vv}[\tau]\\ 
 		& + (R_{uu}[\tau]-R_{vv}[\tau])\bmj + 2R_{uv}[\tau] \bmk.
 	\end{split}
 	\end{equation}
 	The autocovariance function $\gamma_{xx}[\tau]$ takes its values in $\Span{1, \bmj, \bmk}$. It is not symmetric in $\tau$, as the term $R_{uv}[\tau]$ is not symmetric in general. 
%
 	More generally one can define the cross-spectral density between two bivariate stationary random processes $x[t] = u_x[t]+\bmi v_x[t]$ and $y[t] = u_y[t] + \bmi v_y[t]$, where $u_k[t], v_k[t]\in\bbR$, $k=x,y$. Let us denote by $\mathrm{d}X(\nu)$ and $\mathrm{d}Y(\nu)$ their  spectral increments. The \emph{cross-spectral density} between $x$ and $y$ is
 	\begin{equation}\label{eq:CrossSpec}
 		\Gamma_{xy}(\nu)\dnu = \Expe{\dX(\nu)\overline{\mathrm{d}Y(\nu)}} + \Expe{\dX(\nu)\involj{\mathrm{d}Y(\nu)}} \bmj
 	\end{equation}
 	so that their cross-covariance defined as its inverse QFT reads:
	\begin{eqnarray}\label{eq:crosscov}
 		\gamma_{xy}[\tau] = R_{u_xu_y}[\tau] + R_{v_yv_x}[\tau] + (R_{u_yv_x}[\tau]-R_{u_xv_y}[\tau])\bmi \nonumber\\ 
 		 + (R_{u_xu_y}[\tau] - R_{v_yv_x}[\tau])\bmj + (R_{u_yv_x}[\tau]+R_{u_xv_y}[\tau])\bmk.
	\end{eqnarray}
	This quaternion-valued cross-covariance encodes the full statistical information about $x$ and $y$. Eq. (\ref{eq:gamma_xxAutocorrDef}) and (\ref{eq:crosscov}) may sound disappointing at first glance, but there is no simple expression of those equations in terms of usual covariance $\Expe{x[t]\overline{y[t-\tau]}}$ and complementary covariance $\Expe{x[t]y[t-\tau]}$. However the following Wiener-Khintchine like theorem directly connects $x[t]$ and $y[t]$ to the cross-spectral density (\ref{eq:CrossSpec}).


\begin{theorem}
	Let $x$ and $y$ be two stationary random bivariate signals. Then
	\begin{eqnarray}
		\Expe{\dX(\nu)\overline{\mathrm{d}Y(\nu)}} = \displaystyle{\sum_{\tau=-\infty}^{+\infty} \Expe{x[t]e^{-\bmj2\pi\nu\tau}\overline{y[t-\tau]}}} \label{eq:WkTheoremS}\\
		\Expe{\dX(\nu)\involj{\mathrm{d}Y(\nu)}} = \displaystyle{\sum_{\tau=-\infty}^{+\infty} \Expe{x[t]e^{-\bmj2\pi\nu\tau}\involj{y[t-\tau]}} } \label{eq:WkTheoremV}
	\end{eqnarray}
	In the special case $y[t] = x[t]$ one has
	\begin{eqnarray}
		\mathcal{S}(\Gamma_{xx}(\nu)) &= \displaystyle{\sum_{\tau=-\infty}^{+\infty} \Expe{x[t]e^{-\bmj2\pi\nu\tau}\overline{x[t-\tau]}}} \label{eq:WkTheoremS}\\
		\mathcal{V}(\Gamma_{xx}(\nu)) &= \displaystyle{\sum_{\tau=-\infty}^{+\infty} \Expe{x[t]e^{-\bmj2\pi\nu\tau}\involj{x[t-\tau]}} }\bmj \label{eq:WkTheoremV}
	\end{eqnarray}
\end{theorem}
	\begin{proof}[Sketch of proof]
		We start by developing both sides using expressions of $\dX(\nu)$ et $\mathrm{d}Y(\nu)$ in terms of $\mathrm{d}U_k(\nu), \mathrm{d}V_k(\nu)$, $k=x,y$ and $x[t], y[t]$ in terms of $u_k[t], v_k[t]$, $k=x,y$. Then usual rules of quaternion calculus (\emph{e.g.} $\overline{q_1q_2} = \overline{q_2}\:\overline{q_1}$ for $q_1, q_2 \in \bbH$; $\bmi q = \overline{q}\bmi$ if $q \in \bbCj$) permit to simplify both sides. Standard Wiener-Khintchine theorems for real signals lead to the result.
	\end{proof}
%
	Let us note finally the following property.
	\begin{proposition}[Autocorrelation of a sum of independent signals]\label{prop:autocorrIndependent}
		If $x$ and $w$ are independent, $\bbCi$-valued, stationary processes then 
		\begin{equation}
			\gamma_{x+y,x+y}[\tau] = \gamma_{xx}[\tau] + \gamma_{yy}[\tau] 
		\end{equation}
	\end{proposition}
	\begin{proof}
		By direct calculation.
	\end{proof}

	Proposition 1 is a desirable result, which permits to manipulate quaternion valued autocovariance functions like standard autocovariance functions. Note that this result applies with spectral densities of independent signals $x$ and $y$ as well: the spectral density of $x+y$ is the sum of their spectral densities. 



	\subsection{Spectral density and Stokes parameters}\label{ssec:spectralDensityStokes}
	The spectral density $\Gamma_{xx}(\nu)$ is directly related to Stokes parameters, which are fundamental quantities used in polarization of electromagnetic waves \cite{Stokes1852, born2000principles}. Stokes parameters are given by \cite{schreier2010statistical,born2000principles}
	\begin{align}
		S_0(\nu) &= P_{uu}(\nu) + P_{vv}(\nu), \label{eq:S0def}\\
		S_1(\nu) &= P_{uu}(\nu) - P_{vv}(\nu),\\
		S_2(\nu) &= 2\mathfrak{R}\left\lbrace P_{uv}(\nu)\right\rbrace,\\
		S_3(\nu) &= 2\mathfrak{I}_{\bmj}\left\lbrace P_{uv}(\nu)\right\rbrace \label{eq:S3def},
	\end{align}
	where we have introduced the usual spectral densities of $u$ and $v$, $P_{uu}$ and $P_{vv}$, as well as the usual cross-spectral density $P_{uv}$.
	\begin{theorem}
		Let $\Gamma_{xx}(\nu)$ be defined by theorem \ref{thm:spectralRep}. It can be re-expressed like
		\begin{equation}\label{eq:spectralDensStokes}
			\Gamma_{xx}(\nu) = S_0(\nu) + \bmi S_3(\nu) + \bmj S_1(\nu)  + \bmk S_2(\nu).
		\end{equation}
		where $S_\alpha(\nu)$, $\alpha = 0, 1, 2, 3$ are the \emph{Stokes parameters} of $x$.
	\end{theorem}
	\begin{proof}
	The two terms appearing in theorem \ref{thm:spectralRep} can be expressed in terms of the spectral increments $\mathrm{d}U(\nu)$ and $\mathrm{d}V(\nu)$:
	\begin{equation}
		\begin{split}
			\Expe{\vert\dX(\nu)\vert^2} & = \Expe{\vert\dU(\nu)\vert^2} + \Expe{\vert\dV(\nu)\vert^2} \\
			&= ( P_{uu}(\nu) + P_{vv}(\nu))\dnu,
		\end{split}
	\end{equation}
	and
	\begin{equation}
		\begin{split}
			\Expe{\polarmodj{\dX(\nu)}} & = \Expe{\vert\dU(\nu)\vert^2} - \Expe{\vert\dV(\nu)\vert^2} \\
			&+ 2\Expe{\dU(\nu)\overline{\dV(\nu)}}\bmi\\
			& = (P_{uu}(\nu) - P_{vv}(\nu) + 2P_{uv}(\nu)\bmi)\dnu.
		\end{split}
	\end{equation}
	Then (\ref{eq:spectralDensStokes}) follows from (\ref{eq:S0def}) -- (\ref{eq:S3def}). 
	\end{proof}

	Eq. (\ref{eq:spectralDensStokes}) has a powerful geometric interpretation. Stokes parameters permit to separate the contribution of \emph{polarized} and \emph{unpolarized} components, as discussed next in Section \ref{sec:decompositionPolarUnpolar}. The scalar part $S_0(\nu)$ of $\Gamma_{xx}(\nu)$ is the total power at frequency $\nu$, \ie the sum of the power of the polarized and unpolarized parts. The vector part describes only the polarized part of $x$ at frequency $\nu$.{}

	\subsection{Poincar\'e sphere and degree of polarization}\label{sub:Poincaresphere}

	Fig. \ref{fig:PoincareSphere} depicts the Poincar\'e sphere of polarization states \cite[p. 125]{brosseau1998fundamentals}\cite{Poincare1892a,born2000principles}. At frequency $\nu$, the vector part of $\Gamma_{xx}(\nu)$ -- normalized by its scalar part $S_0(\nu)$ -- identifies a point on this Poincar\'e sphere. The angular coordinates $(2\theta, 2\chi)$ are directly related to the \emph{mean ellipse} properties of the signal, \ie $\theta$ is the mean orientation and $\chi$ is the mean ellipticity. 


	At each frequency, the radius of the Poincar\'e sphere is called the \emph{degree of polarization} $\Phi(\nu)$. Namely, 
	\begin{equation}\label{eq:degreePolarDef}
		\Phi(\nu) = \frac{\sqrt{S_1^2(\nu) + S_2^2(\nu) + S_3^2(\nu)}}{S_0(\nu)} = \frac{\vert \mathcal{V}(\Gamma_{xx}(\nu))\vert}{\mathcal{S}(\Gamma_{xx}(\nu))},
	\end{equation}
	where $\mathcal{V}(\cdot)$ and $\mathcal{S}(\cdot)$ denote the vector and scalar part, respectively. It follows from the definition of $\Phi(\nu)$ that $0 \leq \Phi(\nu) \leq 1$ for all $\nu$. The degree of polarization $\Phi(\nu)$ quantifies the repartition between polarized and unpolarized components. This motivates the following vocabulary: the process $x$ is said to be
	\begin{itemize}
		\item  fully polarized at frequency $\nu$ if $\Phi(\nu) = 1$, 
		\item  unpolarized at frequency $\nu$ if $\Phi(\nu) = 0$,
		\item  partially polarized at frequency $\nu$ if $0 < \Phi(\nu) < 1$.
	\end{itemize}
	The degree of polarization is a quantity of fundamental interest in many fields (see \emph{e.g.} \cite{Kikuchi2001,Shirvany2012}).It is invariant by change of reference frame, making it a robust parameter of interest.
	\begin{figure}
		\centering
		\includegraphics[width=.45\textwidth]{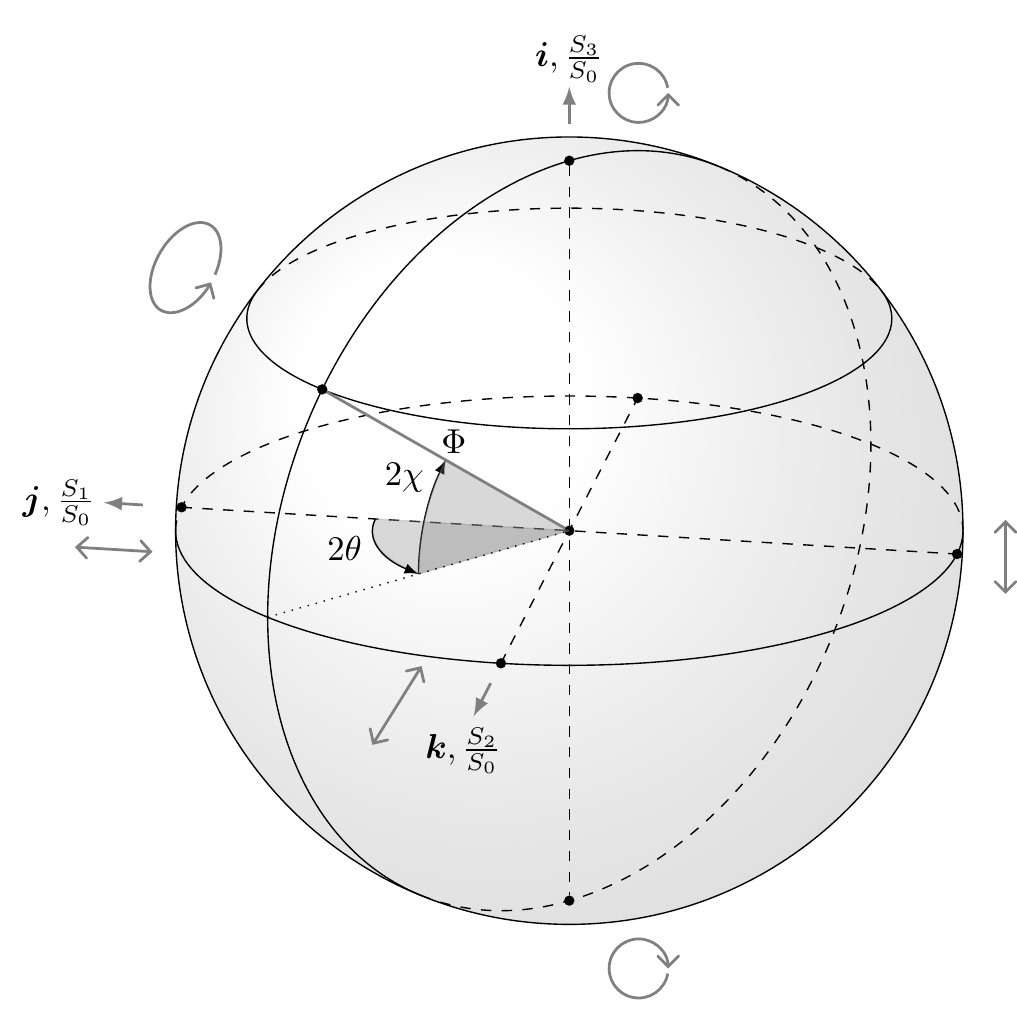}
		\caption{Poincar\'e sphere representation of polarization states. Frequency dependence has been dropped for clarity. A point on the sphere corresponds to a particular polarization state, given by its spherical coordinates $(\Phi, 2\theta, 2\chi)$. }\label{fig:PoincareSphere}
	\end{figure}

	\subsection{Decomposition into polarized and unpolarized parts} \label{sec:decompositionPolarUnpolar} 

	The Unpolarized/Polarized part decomposition (hereafter termed UP decomposition) has been evoked in Section \ref{ssec:spectralDensityStokes}. Now we formally show that any bivariate stationary signal can be uniquely decomposed, as the sum of two uncorrelated processes: a completely polarized process and an unpolarized process. Let us rewrite the spectral density of theorem \ref{thm:spectralRep} as
	\begin{align}
		\Gamma_{xx}(\nu)\dnu &= \left[1-\Phi(\nu)\right]\Expe{\vert\dX(\nu)\vert^2} \nonumber\\
		&+ \left[\Phi(\nu)\Expe{\vert\dX(\nu)\vert^2} + \Expe{\polarmodj{\dX(\nu)}} \bmj\right]\nonumber\\
		&=\Gamma_{xx}^{\tt u}(\nu)\dnu + \Gamma_{xx}^{\tt p}(\nu)\dnu,\label{eq:decompoUnpolarPolar}
	\end{align}
	where the ${\tt u}$ and ${\tt p}$ superscripts stand for unpolarized and polarized part, respectively. The decomposition (\ref{eq:decompoUnpolarPolar}) is unique. Using Stokes parameters, we get
	\begin{equation}
	\begin{split}
		\Gamma_{xx}(\nu) &= \left[1-\Phi(\nu)\right]S_0(\nu) \\
		&+ \left[\Phi(\nu)S_0(\nu) + \bmi S_3(\nu) + \bmj S_1(\nu)  + \bmk S_2(\nu)\right]
	\end{split}
	\end{equation}
	which is the quaternion counterpart of the decomposition given in standard optics textbooks, see \emph{e.g.} \cite[p. 127]{brosseau1998fundamentals}, \cite[p. 551]{born2000principles}. Eq. (37) highlights how the degree of polarization rules the power repartition between the polarized and unpolarized parts of the spectral density. 
	
	 Back to the time-domain, there exists two uncorrelated processes $x_{\tt p}$ and $x_{\tt u}$ such that $x[t] = x_{\tt p}[t] + x_{\tt u}[t]$ and
	\begin{equation}
		\dX(\nu) = \dX_{\tt p}(\nu) + \dX_{\tt u}(\nu),
	\end{equation}
	with $\dX_{\tt p}$ and $\dX_{\tt u}$ the spectral increments corresponding to the polarized and unpolarized part, respectively.

	\subsection{Comparison with previous work}
	\label{sub:comparisonRotarySpectrum}

	\subsubsection{Proper and improper signals}
	The notion of (im)properness of a complex signal has attracted much interest in the signal processing community over the last two decades, see \cite{Picinbono1997b,Adali2011,schreier2010statistical} and references therein. To account for the full second-order statistical structure of a stationary complex signal, one has to consider both the usual autocovariance $R_{xx}[\tau]$ and the complementary covariance $\tilde{R}_{xx}[\tau]$ such that: 
	\begin{align}
	R_{xx}[\tau] &= \Expe{x[t]\overline{x[t-\tau]}}\\
	\tilde{R}_{xx}[\tau] &= \Expe{x[t]x[t-\tau]}
	\end{align}
	Proper signals are characterized by a zero complementary covariance sequence, meaning that a signal $x[t]$ is not correlated with its complex conjugate $\overline{x[t-\tau]}$, for all $\tau$. It follows that
	\begin{equation}\label{eq:propernessConditions}
		\forall \tau,\: R_{uu}[\tau] = R_{vv}[\tau] \text{ and } R_{uv}[-\tau] + R_{uv}[\tau] = 0.
	\end{equation}
	A direct consequence is that the spectral density (\ref{eq:spectralDensStokes}) of a proper signal $x[t]$ reads
	\begin{equation}\label{eq:PSDProperSignal}
		\Gamma_{xx}(\nu) = S_0(\nu) + \bmi S_3(\nu)
	\end{equation}
	as conditions (\ref{eq:propernessConditions}) are equivalent to $S_1(\nu) = S_2(\nu) = 0$ for all $\nu$. Eq. (\ref{eq:PSDProperSignal}) shows that a proper signal is in general partially circularly polarized. This highlights that \emph{polarization} and \emph{properness} of complex random signals are distinct concepts and therefore shall not be confused. 

	\subsubsection{Relation to the rotary spectrum approach}

	To illustrate the relevance of the quaternion-valued spectral density $\Gamma_{xx}(\nu)$ defined in theorem \ref{thm:spectralRep}, we compare it to the well-known rotary spectrum approach \cite{gonella1972rotary,mooers1973technique,walden2013rotary}. This later method decomposes a bivariate signal into a sum of clockwise (CW) and counterclockwise phasors (CCW) -- the so-called rotary components \cite{walden2013rotary}. The determination of the rotary components relies on the usual spectral density $P_{xx}(\nu)$ and complementary spectral density $\tilde{P}_{xx}(\nu)$ defined by standard FTs of $R_{xx}[\tau]$ and $\tilde{R}_{xx}[\tau]$, respectively \cite{schreier2010statistical}.


	For $\nu > 0$, the CW rotary power spectrum is given by $P_{xx}(\nu)$, while $P_{xx}(-\nu)$ gives the CCW rotary power spectrum. The rotary coherence (\ie correlation between CW and CCW components) is controled by $\tilde{P}_{xx}(\nu)$, which is in general complex-valued. 

	The rotary spectra can be expressed in terms of Stokes parameters like \cite[p. 213]{schreier2010statistical}
	\begin{equation}
		P_{xx}(\nu) = S_0(\nu) + S_3(\nu), \quad \tilde{P}_{xx}(\nu) = S_1(\nu) + \bmi S_2(\nu). 
	\end{equation}
	Since $S_0(\nu)$ is even and $S_3(\nu)$ is odd, $P_{xx}(\nu)$ shows no particular symmetry. Moreover, we see that $P_{xx}(\nu)$ combines in one real scalar two very different quantities: $S_0(\nu)$ is related to the total power, and $S_3(\nu)$ gives the (signed) power of the circularly polarized part. This is not surprising since the pair $(P_{xx}(\nu), \tilde{P}_{xx}(\nu))$ was introduced to account for improperness properties of complex-valued signals, not polarization properties. In contrast the use of $\Gamma_{xx}(\nu)$ provides directly the total power information decoupled from the polarization information. 

	The rotary spectrum and the quaternion-valued approach provide equivalent representations. However, the quaternion-valued spectral density $\Gamma_{xx}$ provides a direct interpretation of physical quantities, the Stokes parameters. These parameters appear naturally in the components of $\Gamma_{xx}$. The use of the QFT to study bivariate signals  provides a well-suited framework for a meaningful and rooted in physics ``geometric spectral analysis''. This approach demonstrates that the use of higher dimensional algebra in the definition of the FT permits to avoid \emph{ad hoc} constructions while revealing intrinsic relevant parameters of bivariate signals.

	\subsection{Examples} 
	\label{sub:examplesSpecRep}

	\subsubsection{Deterministic signals}
	In this case the spectral density reads
	\begin{equation}
		\Gamma_{xx}(\nu) = \vert X(\nu)\vert^2 + \polarmodj{X(\nu)}\,\bmj,
	\end{equation}
	where expectations have been dropped out. Moreover one gets
	\begin{equation}
		\Phi(\nu) = \frac{\vert \polarmodj{X(\nu)} \vert}{\vert X(\nu)\vert^2} = \frac{\vert X(\nu)\vert \vert \involj{X(\nu)}\vert}{\vert X(\nu)\vert^2} = 1,
	\end{equation}
	so that \emph{a deterministic signal is always fully polarized at all frequencies for which $X(\nu) \neq 0$.} The other polarization parameters $a, \theta$, $\chi$ or $S_1, S_2, S_3$ can be determined as well, depending on the signal. Let the bivariate monochromatic signal $x$ such that
	\begin{equation}\label{eq:bivariateMonochromaticSignal}
		x[t] = 2ae^{\bmi\theta}(\cos\chi\cos[2\pi\nu_0t] +\bmi\sin\chi\sin[2\pi\nu_0t]),
	\end{equation}
	with $a, \chi, \theta$ the parameters of the elliptic polarization. The QFT of $x$ reads\footnote{For the sake of simplicity, we restrict our analysis to the frequency support $[-1/2, 1/2]$, although $X(\nu)$ is periodic. Note also the use of the Euler polar form (\ref{eq:polarForm}) with a null $\bmj$-phase.}
	\begin{equation}
		X(\nu) = ae^{\bmi\theta}e^{\bmk\chi}\delta(\nu+\nu_0) + ae^{\bmi\theta}e^{-\bmk\chi}\delta(\nu-\nu_0),
	\end{equation}
	which gives the following spectral density $\Gamma_{xx}(\nu)$:
	\begin{equation}
		\Gamma_{xx}(\nu) = \Gamma_{xx}^{\nu_0}\delta(\nu-\nu_0) + \involi{{\Gamma_{xx}^{\nu_0}}}\delta(\nu+\nu_0),
	\end{equation}
	where $\Gamma_{xx}^{\nu_0} = S_0 + \bmi S_3 + \bmj S_1 + \bmk S_2$.
	The autocovariance sequence is then obtained by inverse QFT:
	\begin{equation}
	\begin{split}
		\gamma_{xx}[\tau] &= 2 S_0\cos[2\pi\nu_0\tau]+ \bmj 2S_1 \cos[2\pi\nu_0\tau]\\
		&+2\bmk\left(S_2\cos[2\pi\nu_0\tau] + S_3\sin[2\pi\nu_0\tau]\right).
	\end{split}
	\end{equation}
	The autocovariance is not symmetric. The value of $S_3$ controls the odd contribution, whereas the remaining terms are all even. It is interesting to note that the autocovariance function of a monochromatic signal is even if and only if the signal is linearly polarized. 


	\subsubsection{Bivariate white noise}\label{sssub:improperWN}
	Consider the process $w[t] = u[t] + \bmi v[t]$ where $u, v$ are both real, i.i.d. and jointly second-order stationary with properties:
	\begin{equation}\begin{split}
		&\Expe{u} = \Expe{v} = 0, \\
		 &\Expe{u^2}= \sigma_u^2, \: \Expe{v^2} = \sigma_v^2, \: \Expe{uv} = \rho_{uv}\sigma_u\sigma_v.	
	\end{split}
	\end{equation}
	Since $u, v$ are i.i.d., their autocovariances are
	\begin{equation}
		R_{uu}[\tau] = \sigma_u^2\delta_{\tau, 0}, \: R_{vv}[\tau] = \sigma_v^2\delta_{\tau, 0}, \: R_{uv}[\tau] = \rho_{uv}\sigma_u\sigma_v\delta_{\tau, 0}.
	\end{equation}
	which yields the autocorrelation of $x$ using (\ref{eq:gamma_xxAutocorrDef})
	\begin{equation}
		\gamma_{ww}[\tau] = \left[\sigma_u^2+\sigma_v^2 + \bmj (\sigma_u^2-\sigma_v^2) + 2\bmk \rho_{uv}\sigma_u\sigma_v\right]\delta_{\tau, 0}.
	\end{equation}
	The spectral density is obtained by QFT:
	\begin{equation}\label{eq:spectralDensWhiteNoise}
		\Gamma_{ww}(\nu) = \sigma_u^2+\sigma_v^2 + \bmj (\sigma_u^2-\sigma_v^2) + 2\bmk \rho_{uv}\sigma_u\sigma_v.
	\end{equation}
	This spectral density is constant. It has no $\bmi$-component, so that $S_3(\nu)=0$ for all $\nu$. As a consequence, an \emph{i.i.d. arbitrary second-order stationary bivariate or complex white noise shows no ellipticity. It is always either unpolarized, or linearly polarized (fully or partially). The polarization properties are identical at all frequencies}.



	The polarization degree defined by (\ref{eq:degreePolarDef}) is:
	\begin{equation}\label{eq:polarizationDegreeImproperWN}
		\Phi = \frac{\sqrt{(\sigma_u^2-\sigma_v^2)^2+ 4\rho_{uv}^2\sigma_u^2\sigma_v^2}}{\sigma_u^2+\sigma_v^2},
	\end{equation}
	where we see that $x[t]$ is unpolarized at all frequencies iff $\sigma_u =\sigma_v$ and $\rho_{uv} = 0$. When $\Phi \neq 0$, the angle $\theta$ of the linear polarization is given by $\theta = 0$ if $\rho_{uv} = 0$ and by 
	\begin{equation}\label{eq:polarizationAngleImproperWN}
	\begin{cases}
		\theta &= \frac{1}{2}\mathrm{atan2}\left[\frac{2\rho_{uv}\sigma_u\sigma_v}{(\sigma_u^2-\sigma_v^2)}\right], \text{ if } \sigma_u \neq \sigma_v\\
		\theta &= \pi/4, \text{ if } \sigma_u = \sigma_v
	\end{cases}
	\end{equation}
	when $\rho_{uv} \neq 0$ and where $\mathrm{atan2}$ denotes the four-quadrant inverse tangent.


	The UP decomposition of bivariate white noise gives a simple procedure to simulate bivariate white noise with desired polarization properties. Let $ 0 \leq \Phi\leq 1$ be the desired degree of polarization, $\theta \in [-\pi/2, \pi/2]$ the orientation angle and $S_0 > 0$ the total intensity. Let $w^{\texttt{u}}[t]$ be an unpolarized white noise, $\bbCi$-valued, such that $R_{ w^{\texttt{u}} w^{\texttt{u}}}[\tau] = \delta_{\tau, 0}$. Let $w^{\texttt{p}}[t]$ be a real-valued white noise sequence of unit variance. Assume further that $w^{\texttt{u}}[t]$ and $w^{\texttt{p}}[t]$ are independent. Then the white noise $w[t]$ constructed as
	\begin{equation}\label{eq:UPdecompWN}
		w[t] = \sqrt{S_0}\left(\sqrt{1-\Phi}w^{\texttt{u}}[t] + \sqrt{\Phi}\exp(\bmi \theta)w^{\texttt{p}}[t]\right)
	\end{equation}
	has exactly the desired polarization properties. 



	\subsubsection{Bivariate monochromatic signal in white noise}\label{ssub:bivMonoImWN}

	Consider the signal $y[t] = x[t] + w[t]$, where $x[t]$ is a bivariate monochromatic signal and $w[t]$ is a bivariate white noise (i.i.d. sequence). Assume moreover that $x[t]$ and $w[t]$ are independent. From Proposition \ref{prop:autocorrIndependent} the spectral density of $y$ is directly given by the sum of their spectral densities:
	\begin{equation}\label{eq:GammayyGammaxxGammaww}
		\Gamma_{yy}(\omega) = \Gamma_{xx}^{\nu_0}\delta(\nu-\nu_0) + \involi{{\Gamma_{xx}^{\nu_0}}}\delta(\nu+\nu_0) + \Gamma_{ww}(\nu)
	\end{equation}
	where $\Gamma_{ww}(\nu)$ is given by (\ref{eq:spectralDensWhiteNoise}) and $\Gamma_{xx}^{\nu_0} = S_{0, x} + \bmi S_{3, x} + \bmj S_{1, x} + \bmk S_{2, x}$;  the $S_{\alpha, x}$ denote the Stokes parameters of $x$. At $\nu = \nu_0$ the spectral density writes $\Gamma_{yy}^{\nu_0} = S_{0, y} + \bmi S_{3, y} + \bmj S_{1, y} + \bmk S_{2, y}$ where
	\begin{equation}\label{eq:explicitGammayyGammaxxGammaww}
		\begin{split}
			S_{0, y} &= S_{0, x} + \sigma_u^2+\sigma_v^2, \quad S_{1, y}  = S_{1, x}  + \sigma_u^2-\sigma_v^2 \\
			S_{2, y}  &= S_{2, x}  + 2\rho_{uv}\sigma_u\sigma_v, \quad S_{3, y} = S_{3, x}.
		\end{split}
	\end{equation}
	First, we see that $S_1$ and $S_2$ parameters are mixing polarization properties of $x$ and $w$. Since $S_3$ is not modified in presence of i.i.d. white noise, only the direction of polarization changes, not the ellipticity. The output degree of polarization (at $\nu = \nu_0$) takes a simple form when the noise is unpolarized:
	\begin{equation}
		\Phi_y = \frac{\sqrt{S_{1, x}^2 + S_{2, x}^2 + S_{3, x}^2}}{S_{0, x} + \sigma_u^2 + \sigma_v^2} = \frac{\mathrm{SNR}}{\mathrm{SNR} + 1} \Phi_x \leq \Phi_x
	\end{equation}
	where $\mathrm{SNR} = S_{0, x}/(\sigma^2_u + \sigma_v^2)$ is the signal-to-noise ratio (SNR). The degree of polarization decreases with the SNR.

\section{Spectral density estimation} 
\label{sec:spectral_density_estimation}
	
	We propose two nonparametric spectral density estimation methods, and derive their properties. In particular, we investigate the problem of estimating the degree of polarization. In the remainder of this paper, we consider a bivariate stationary signal $x[t] = u[t] + \bmi v[t]$ consisting in $N$ samples such that $t = 0, 1, \ldots N-1$ and with sampling size $\Delta = 1$.

	\subsection{A naive spectral estimator: the polarization periodogram} 
	\label{sub:polarization_periodogram}

	The first basic spectral density estimator is the \emph{polarization periodogram} $\hat{\Gamma}_{xx}^{(p)}(\nu)$. The underlying rationale is very close to the derivation of the usual periodogram. One starts by computing an estimator $\hat{\gamma}_{xx}^{(p)}[\tau]$ of the autocovariance sequence $\gamma_{xx}[\tau]$. It is done by combining usual (biased) estimators of auto- and cross-covariance sequences:
	\begin{equation}\label{eq:crossCorrEstimator}
		\hat{R}_{uv}^{(p)}[\tau] = \frac{1}{N}\sum_{t=1}^{N-\tau}u[t+\tau]v[t], \: \tau = 0, 1, \ldots N-1
	\end{equation}
	where $\hat{R}_{uv}^{(p)}[\tau] = \hat{R}_{vu}^{(p)}[-\tau]$ for $\tau = -1, \ldots, -(N-1)$ and $\hat{R}_{uv}^{(p)}[\tau] = 0$ for $\vert \tau\vert \geq N$. Auto-covariance estimators follow from (\ref{eq:crossCorrEstimator}). Then taking the QFT of $\hat{\gamma}_{xx}^{(p)}[\tau]$ given by (\ref{eq:gamma_xxAutocorrDef}) yields the \emph{polarization periodogram}, which reads (after simplification)
	\begin{equation}\label{eq:polarizationPeriodogram}
	\begin{split}
		\hat{\Gamma}_{xx}^{(p)}(\nu) &= N^{-1}\left\vert \sum_{t=1}^N x[t]e^{-\bmj2\pi\nu t}\right\vert^2 \\&+  N^{-1}\polarmodj{\sum_{t=1}^N x[t]e^{-\bmj2\pi\nu t}}\bmj
	\end{split}
	\end{equation}

	Alike the classical periodogram, this estimator is a biased, inconsistent estimator of the spectral density $\Gamma_{xx}(\nu)$. One has
	\begin{equation}
		\Expe{\hat{\Gamma}_{xx}^{(p)}(\nu)}= \intnu \mathcal{F}_N(\nu-\nu')\Gamma_{xx}(\nu')\mathrm{d}\nu'
	\end{equation}
	where $\mathcal{F}_N(\nu) = \sin^2(\pi N\nu)/[N\sin^2(\pi\nu)]$ is the Fej\'er kernel. It follows that the polarization periodogram is only asymptotically unbiased.  Note however that in the case of white noise, as $\Gamma_{ww}(\nu)$ is constant, the polarization periodogram is unbiased for any $N$. This is a bivariate counterpart of a classical result, see \emph{e.g.} \cite[p. 202]{percival1993spectral}.

	Alike in standard spectral analysis \cite{percival1993spectral}, data tapers are to be employed to produce a direct spectral estimator with better bias properties than the naive polarization periodogram.


	\subsection{Multitapering}
	\label{sub:multitaper}
	The multitaper method is a well established technique \cite{percival1993spectral,Thomson1982,Walden2000} which produces a spectral density estimate with reduced variance, while maintaining good bias properties. The basic idea is to compute a series of $K$ \emph{direct estimators} $\hat{\Gamma}_{xx}^{k}(\nu)$, $k = 0, 1, \ldots K-1$ that are approximately uncorrelated \cite{percival1993spectral}. The $k$-th spectral estimator reads
	\begin{equation}
	\begin{split}
		\hat{\Gamma}_{xx}^{k}(\nu) &= \left\vert \sum_{t=1}^N h_k[t]x[t]e^{-\bmj2\pi\nu t}\right\vert^2 \\&+ \polarmodj{\sum_{t=1}^N h_k[t]x[t]e^{-\bmj2\pi\nu t}}\bmj,
	\end{split}
	\end{equation}
	where the $h_k$'s are real-valued sequences of size $N$. They are normalized $(\sum_{t=0}^{N-1} h_k[t]^2 = 1)$ and orthogonal 
	\begin{equation}
		\sum_{t = 0}^{N-1} h_k[t]h_{k'}[t] = \delta_{k, k'}.
	\end{equation}
	Functions satisfying these conditions together with good leakage properties are for instance the Slepian tapers \cite{slepian1978prolate}. Then the multitaper estimate is obtained by averaging:
	\begin{equation}\label{eq:spectralDensityMT}
		\hat{\Gamma}_{xx}^{(\mathrm{mt})}(\nu) = \frac{1}{K}\sum_{k=1}^{K}\hat{\Gamma}_{xx}^{k}(\nu).
	\end{equation}
	In practice the value of $K$ often satisfies $K < 10$ to maintain a good spectral resolution \cite{percival1993spectral}.

	\subsection{Estimation of the degree of polarization} 
	\label{sub:degPolarEstimation}

	\subsubsection{Theoretical properties}
	The estimation of the degree of polarization (\ref{eq:degreePolarDef}) has attracted interest in the signal processing community \cite{Medkour2008,SantalladelRio2006} in relation to many fields \cite{Kikuchi2001,Shirvany2012}.
	A naive estimator (at frequencies where the polarization periodogram is nonzero) based on the polarization periodogram would be trivial since:
	\begin{equation}\label{eq:degreePolarEstimatePeriodogram}
		\hat{\Phi}^{(p)}(\nu) = \frac{\vert \mathcal{V}(\hat{\Gamma}_{xx}^{(p)}(\nu))\vert}{\mathcal{S}(\hat{\Gamma}_{xx}^{(p)}(\nu))} = 1,
	\end{equation}
	which is systematically biased, except for frequencies where $x[t]$ is fully polarized. In a situation where $M$ approximately uncorrelated estimates of the spectral density are available (having multiple realizations of $x$ or using a multitaper estimate (\ref{eq:spectralDensityMT}), in which case $M = K$) one can form a new estimate of the degree of polarization as
	\begin{equation}
		\hat{\Phi}^{M}(\nu) = \frac{\vert \sum_{m=1}^{M}\mathcal{V}(\hat{\Gamma}_{xx}^{m}(\nu))\vert}{\sum_{m=1}^{M}\mathcal{S}(\hat{\Gamma}_{xx}^{m}(\nu))},\label{eq:polarizationDegreeEstimate}
	\end{equation}
	which is a better estimator of $\Phi$ than (\ref{eq:degreePolarEstimatePeriodogram}). Medkour and Walden\footnote{Their approach is based on spectral matrices rather than the quaternion-valued spectral density introduced here. However this does not change the nature of their results, since definitions of the degree of polarization are identical.} \cite{Medkour2008} studied theoretically this estimator in a Gaussian setting and showed that it is unbiased in the limit $M \rightarrow \infty$. 

	\subsubsection{Numerical study}
	We propose to numerically study its performances. To avoid spectral blurring effects, we consider the (Gaussian) white noise case since in that case the polarization periodogram is an unbiased estimator of the spectral density:
	\begin{equation}
		\Expe{\hat{\Gamma}_{ww}^{(p)}(\nu)}= \Gamma_{ww}(\nu).
	\end{equation}
	The UP decomposition (\ref{eq:UPdecompWN}) of bivariate white noise allows to generate white noise with prescribed polarization properties. We fix $S_0 = 1$ and generate $M$ independent bivariate white Gaussian noise sequences, leading to $M$ independent periodogram estimates of the spectral density (\ref{eq:polarizationPeriodogram}). The degree of polarization is then estimated by (\ref{eq:polarizationDegreeEstimate}). The bias is estimated by Monte-Carlo simulation. 

	Fig. \ref{fig:degPolBias} depicts the bias in the estimation of the degree of polarization, for $M = 1, 2, 5, 10, 20, 50$ and $500$. Given $M$, the bias increases as the true degree of polarization goes to 0. The bias decreases with larger values of $M$, and becomes negligible for $M \rightarrow \infty$. Note that for typical values of $M < 10$ admissible in multitaper estimation, the bias remains significant for $\Phi \simeq 0.6$. Our results agree with those of \cite{Medkour2008}.

	\begin{figure}
		\centering
		\includegraphics[width=.4\textwidth]{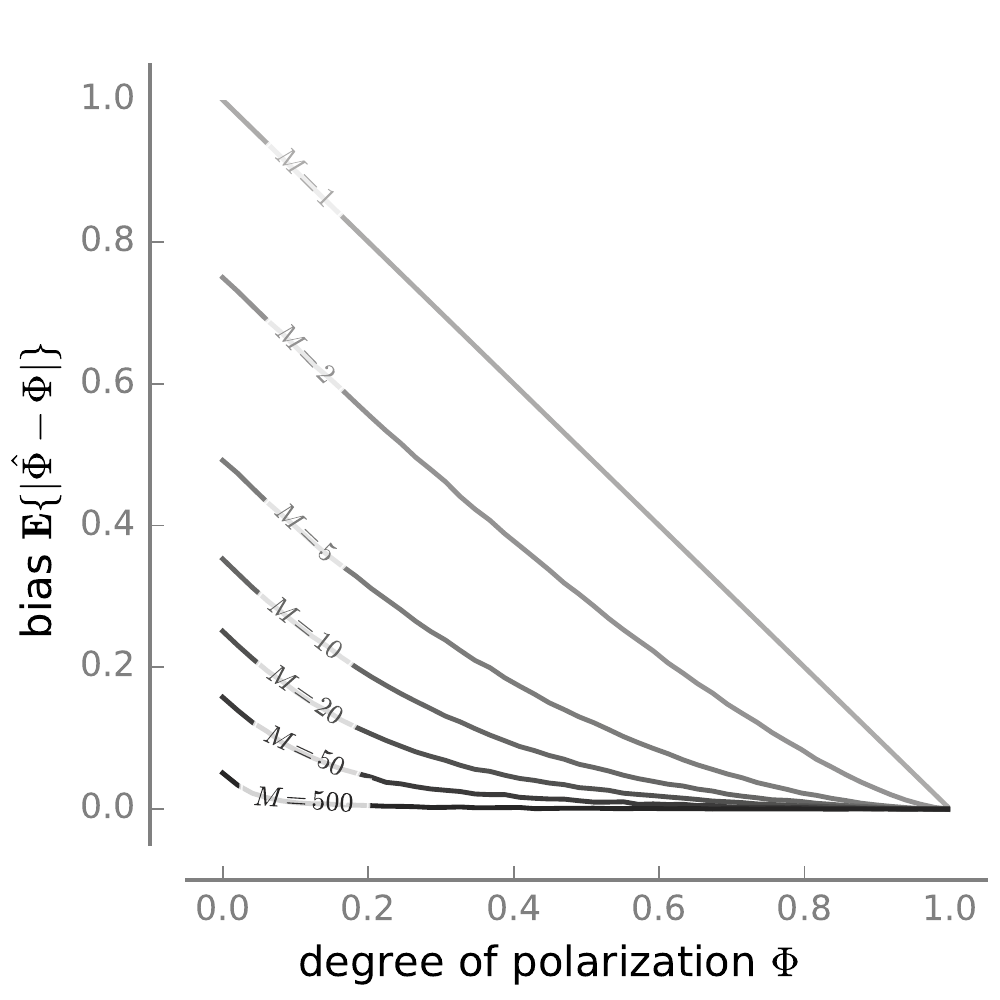}\caption{Estimation bias of the degree of polarization obtained by averaging $M$ independent polarization periodogram estimates. The bias is smaller for large values of $M$ and a degree of polarization close to unity.}
		\label{fig:degPolBias}
	\end{figure}


	

	\subsection{Summary}
	Spectral density estimation of bivariate signals suffers from two biases. As in the univariate case, the naive polarization periodogram $\hat{\Gamma}^{(p)}_{xx}(\nu)$ is a biased estimator of the spectral density since the signal is of finite length $N$. Well-known multitapering techniques can be adapted to handle this bias and reduce variance. 

	The second source of bias relates to the estimation of polarization properties. Precisely, a key quantity is the degree of polarization, as it relates the ratio between polarized and unpolarized parts of the signal. Spectra of polarization attributes are more difficult to obtain than simple power spectra. They require the observation of many realizations to reach a good accuracy.


\section{Numerical validation} 
\label{sec:examples}

We consider the synthetic bivariate signal of Section \ref{ssub:bivMonoImWN} $y[t] = x[t] + w[t]$, where $x$ is a bivariate monochromatic signal defined by (\ref{eq:bivariateMonochromaticSignal}) and where $w$ is a bivariate white Gaussian noise given in (\ref{eq:UPdecompWN}). All signals are of length $N = 1024$. We consider positive frequencies only, as negative frequencies can be obtained by symmetry (\ref{eq:isymmetryGammaxx}). 

The frequency of the monochromatic signal $x$ is set to $\nu_0 = 128/N = 0.125$. The signal $x$ has Stokes parameters at frequency $\nu_0$: $S_{0, x}(\nu_0) = 1$, $s_{1, x}(\nu_0) = S_{1, x}(\nu_0)/S_{0, x}(\nu_0) = -0.354$, $s_{2, x}(\nu_0) = S_{2 x}(\nu_0)/S_{0, x}(\nu_0) =  -0.612$, $s_{3, x}(\nu_0) = S_{3, x}(\nu_0)/S_{0, x}(\nu_0) = 0.707$ and $\Phi_x(\nu_0) = 1$ since $x$ is deterministic. For $\nu \neq \nu_0$, all Stokes parameters $S_{\alpha, x}$, $\alpha = 0, ..., 3$ are zero. Equivalently, $x$ is defined by $a_x = 1$, $\theta_x = -\pi/3$ and $\chi_x = \pi/8$ using the polar form (\ref{eq:bivariateMonochromaticSignal}). 

The white noise signal $w$ has constant-frequency Stokes parameters (see Section \ref{sssub:improperWN}) such that $S_{0, w}(\nu) = 10/N$, $\Phi_w(\nu) = 0.2$, $\theta_w(\nu) = \pi/8$ or equivalently, $s_{1, w}(\nu) = 0.141$ and $s_{2, w}(\nu) = 0.141$. Since $w$ is a white noise sequence (and thus has no memory), $s_{3, w}(\nu) = 0$ for any $\nu$. 

The spectral description of $y[t] = x[t] + w[t]$ was derived explicitely in Section \ref{ssub:bivMonoImWN}. Using expressions (\ref{eq:GammayyGammaxxGammaww}) and (\ref{eq:explicitGammayyGammaxxGammaww}) we see that the resulting Stokes parameters at $\nu \neq \nu_0$ are those of $w$. At frequency $\nu_0$ we have
\begin{equation}
	\begin{split}
		S_{0, y}(\nu_0) &= S_{0, x}(\nu_0) + S_{0, w}(\nu_0), \quad S_{1, y}(\nu_0) = S_{1, x}(\nu_0) + S_{1, w}(\nu_0)\\\
		S_{2, y}(\nu_0) &= S_{2, x}(\nu_0) + S_{2, w}(\nu_0), \quad S_{3, y}(\nu_0) = S_{3, x}(\nu_0).
	\end{split}
\end{equation}
Normalized Stokes parameters and polarization degree then read for frequency $\nu_0$
\begin{equation}
\begin{split}
	s_{1, y}(\nu_0) = -0.349, \:  s_{2, y}(\nu_0) = -0.605, \\
	s_{3, y}(\nu_0) = 0.707, \: \Phi_y(\nu_0) = 0.989.
\end{split}
\end{equation}
Due to the polarization properties of $w$, the polarization properties at $\nu_0$ of $y$ are not the same as those of $x$. 

To investigate the estimation of the spectral density of $y$, we have generated $M = 20$ realizations. For each realization, we compute the polarization periodogram and the multitaper estimate using $K = 5$ Slepian tapers. To reduce the bias of the degree of polarization estimate, those estimates are averaged to produce a polarization periodogram estimate and a multitaper estimate. This also reduces bias on normalized Stokes parameters estimates, as one expects the estimates of $s_1, s_2, s_3$ to be biased as they depend on the value of the degree of polarization. 

Fig. \ref{fig:monoDSP} shows simulation results, and one realization of the process $y$ is shown. As expected the (averaged) multitaper estimate has less variance than the (averaged) polarization periodogram. Thin lines indicate the theoretical values of the components of the spectral density, showing accurate estimation of the spectral quantities.

	\begin{figure*}
		\centering
		\includegraphics[width=\textwidth]{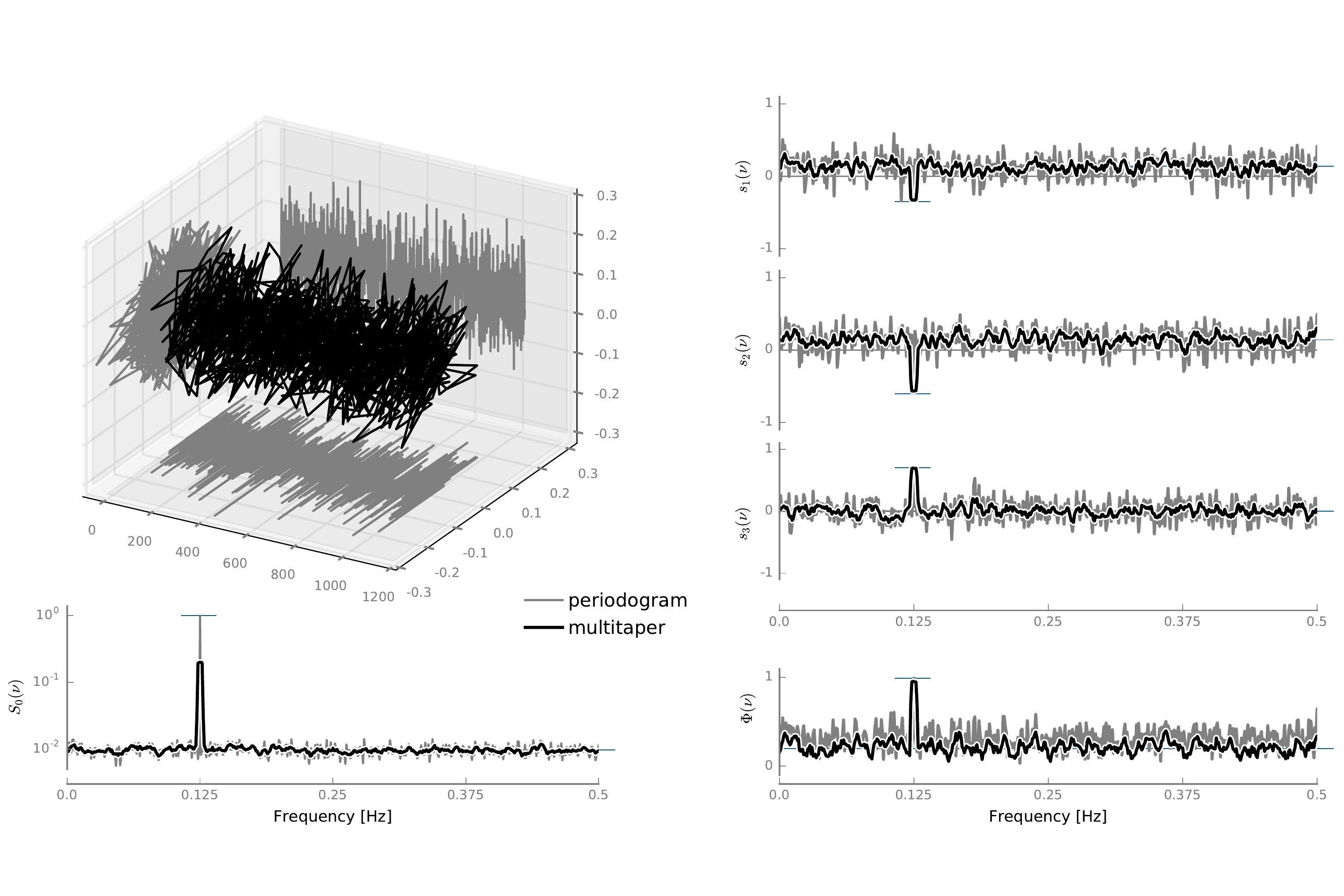}\caption{Spectral density estimation of the $y = x + w$ signal, where $x$ is a monochromatic bivariate signal and $w$ is a bivariate white Gaussian noise. Two estimates are presented, the averaged polarization periodogram and averaged multitaper estimate (computed with $K =5$ Slepian tapers). They are constructed by averaging single estimates obtained via $M = 20$ independent observations of the process $y$. Thin lines indicate the theoretical values of intensity parameter $S_0(\nu)$, normalized Stokes parameters $s_\alpha(\nu) = S_\alpha(\nu)/S_0(\nu)$, $\alpha=1, 2, 3$ and degree of polarization $\Phi(\nu)$. }\label{fig:monoDSP}
	\end{figure*}

\section{Conclusion} 
\label{sec:conclusion}

This paper provides a powerful and relevant framework for an interpretable and efficient spectral analysis of stationary bivariate processes. The richness of the quaternion algebra permits a fruitful interplay between mathematical tools and physical features. Using the QFT, we have introduced the quaternion valued-spectral representation of a bivariate stationary random signal. As a result, the quaternion-valued spectral density is defined, which leads naturally to a spectral analysis in terms of Stokes parameters. It permits a direct interpretation of both power and polarization features of the signal. Simple theoretical examples demonstrate the efficiency of the approach. 
Nonparametric spectral density estimation has been investigated. The limitations of the proposed tools can be studied using standard techniques of univariate spectrum analysis. Moreover, we have stressed the issue raised by the degree of polarization estimation and polarization attributes. These key quantities are relevant to the analysis of bivariate signals but require more care than standard spectral analysis of univariate signals.
Our approach is very generic, and generalizes the standard toolbox of spectral analysis to bivariate stationary signals. It paves the way to new developments in the simulation, estimation and filtering of bivariate signals.

\appendices
\section{Proof of the spectral representation theorem \ref{thm:spectralRep}}\label{app:spectralRepProof}
The proof is divided in two parts, for clarity. 
	\paragraph{Existence}
	Let $x[t] = u[t] + \bmi v[t]$, where $u[t], v[t]$ are real-valued, zero-mean, harmonizable stationary processes. These real processes admit a spectral representation, such that
	\begin{equation}
	\begin{split}
	 	u[t] &= \intnu \mathrm{d}U(\nu)\exp(\bmj2\pi\nu t), \\
		v[t] &= \intnu \mathrm{d}V(\nu)\exp(\bmj2\pi\nu t),
	\end{split}
	\end{equation} 
	since the QFT applied to $\bbCj$-valued signals is equivalent to the usual Fourier transform. By linearity of the QFT, the spectral increments of $x$ are $\dX(\nu) = \mathrm{d}U(\nu) + \bmi \mathrm{d}V(\nu)$, so that
	\begin{equation}
		x[t] = \intnu \mathrm{d}X(\nu)\exp(\bmj2\pi\nu t)
	\end{equation}
	holds for all $t$ in the mean-square sense.

	\paragraph{Properties of the spectral increments}
	The properties of the spectral increments $dX(\nu)$ are a direct consequence of the properties of the spectral increments of $u$ and $v$, respectively. If $x$ is assumed zero-mean stationary,
	\begin{equation}
	\begin{split}
		\forall t, \Expe{x[t]} &= \intnu \Expe{\mathrm{d}X(\nu)}\exp(\bmj2\pi\nu t) = 0\\
		& \Rightarrow \Expe{\mathrm{d}X(\nu)} = \Expe{x[t]} = 0.
	\end{split}
	\end{equation}
	Turning to the second-order properties of the spectral increments, let us consider the spectral representation of $u$ and $v$. Second-order stationarity implies that (see \cite{priestley1981spectral} for details)
	\begin{equation}\label{app:uvspectralProp}
		\forall \nu \neq \nu', \begin{cases} \Expe{\mathrm{d}{U}(\nu)\overline{\mathrm{d}{U}(\nu')}} &= 0\\
		\Expe{\mathrm{d}{V}(\nu)\overline{\mathrm{d}{V}(\nu')}} &= 0
		\end{cases}
	\end{equation}
	and autocorrelation functions of $u, v$ read
	\begin{align}
	 	\Expe{u[t]u[t-\tau]} &= \intnu \Expe{\vert \mathrm{d}{U}(\nu)\vert^2}e^{\bmj2\pi\nu\tau}\label{app:autocovu},\\
	 	 \Expe{v[t]v[t-\tau]} &= \intnu \Expe{\vert \mathrm{d}{V}(\nu)\vert^2}e^{\bmj2\pi\nu\tau}.\label{app:autocovv}
	\end{align} 
	The quantity $\Expe{\vert \mathrm{d}{U}(\nu)\vert^2}$ is interpreted as the spectral density $P_{uu}$ of $u$ times $\dnu$. The same result holds for $v$.

	To fully characterize the spectral increments of $x$, we also need the covariance between the spectral increments of $u$ and $v$. Since $u$ and $v$ are jointly second-order stationary,
	\begin{equation}\label{app:uvspectralProp2}
		\forall \nu \neq \nu', \Expe{\mathrm{d}{U}(\nu)\overline{\mathrm{d}{V}(\nu')}} = 0,
	\end{equation}
	and the cross-correlation function reads
	\begin{equation}\label{app:crosscov}
	 	\Expe{u[t]v[t-\tau]} = \intnu\Expe{\mathrm{d}{U}(\nu)\overline{\mathrm{d}{V}(\nu)}}e^{\bmj2\pi\nu\tau}.
	 \end{equation} 
	As a result we have from (\ref{app:uvspectralProp}) and (\ref{app:uvspectralProp2}):
	\begin{align}
		\forall \nu \neq \nu', &\: \Expe{\dX(\nu)\overline{\dX(\nu')}} = 0,\\
		\forall \nu \neq \nu', &\: \Expe{\dX(\nu)\involj{\dX(\nu')}} = 0.
	\end{align}
	When $\nu' = \nu$, the properties are summarized by the spectral density $\Gamma_{xx}(\nu)$
	\begin{equation}
		\Expe{\vert\dX(\nu)\vert^2} + \Expe{\dX(\nu)\involj{\dX(\nu)}}\bmj  = \Gamma_{xx}(\nu)\domega
	\end{equation}
	which separates in quaternion algebra the information contained in the two moments of the spectral increments. This theorem holds also for  quaternion-valued stationary signals by simply adapting the proof. As a corollary, combining (\ref{app:autocovu}), (\ref{app:autocovv}) and (\ref{app:crosscov}) for $\tau = 0$ yields (\ref{eq:powerSpectralDensity}).

\bibliography{references.bib}
\bibliographystyle{IEEEtran}

\end{document}